\documentclass[final]{ws-ijgmmp}

\usepackage[utf8]{inputenc}

\begin{document}

\markboth{J. Muñoz-Masqué, J. A. Vallejo}
{Supersymmetric sigma models and harmonic superfunctions}

%\catchline{}{}{}{}{}
%
%%%%%%%%%%%%%%%%%%%%%%%%%%%%%%%%%%%%%%%%%%%%%%%%%%%%%%%%%%%%%%%%%%%%

\title{SUPERSYMMETRIC SIGMA MODELS AND HARMONIC SUPERFUNCTIONS}

\author{JAIME MUÑOZ-MASQUÉ}

\address{Instituto de F\'{\i}sica Aplicada\\
Consejo Superior de Investigaciones Cient\'{\i}ficas\\
Serrano 141 CP28006 Madrid (Spain)\\
\email{jaime@iec.csic.es} }

\author{JOSÉ ANTONIO VALLEJO}

\address{Facultad de Ciencias\\
Universidad Aut\'{o}noma de San Luis Potos\'{\i} \\
Salvador Nava s/n CP78290 San Luis Potosí (México)\\
jvallejo@fciencias.uaslp.mx }

\maketitle

%\begin{history}
%\received{(Day Month Year)}
%\revised{(Day Month Year)}
%\end{history}

\begin{abstract}
We study the relation between the Laplacian associated to an odd metric
on a supermanifold and harmonic superfunctions, through the application
of the calculus of variations to a supersymmetric sigma model.
\end{abstract}


\section{Introduction}

In the last years (since the apparition of the Batalin-Vilkovisky
quantization method) there has been a certain amount of interest on
Laplace operators of odd parity on supermanifolds, mainly in connection
with their use in the quantization of non Abelian gauge theories (BRST
symmetries, see \cite{Sch93,Aoy94,KV02,KV041,BatBer08}.) These operators
have been related to odd Poisson structures and odd divergences, modular
classes, etc. (see \cite{Wit90,KosMon02,KV042}.) But we want to focus
here on the ``Riemannian geometric'' side of the problem, and the
relation between Laplace's equation and harmonic functions (for a physical 
interpretation and relation to string theory, see \cite{MV09}.) In the classical
setting, harmonic functions are obtained as follows.
Consider two fixed Riemannian manifolds $(M,g)$ and $(N,h)$. Given
a mapping $\phi \colon (M,g)\to (N,h)$, we can take the pullback
of the metric $h$ by $\phi$ and then raise an index with the inverse
metric $g^{-1}$. If $A=g^{-1}\cdot \phi ^{*}h$ is the resulting bundle
endomorphism $A\colon TM\to TM$, its trace $trA$ is a function
on $M$ which can be integrated (with respect to the Riemannian volume
of $g$) to give the functional
\[
\phi\mapsto \int_{M}\left( trA\right) vol_{g}.
\]
Critical points of this functional are called harmonic mappings (see 
\cite{Eel64,Eel78,Nis02}.) In
the particular case $N=\mathbb{R}$ they are also called harmonic
functions, and the Euler-Lagrange equations for the above functional
are precisely\[
\triangle\phi=0,\]
so that harmonic functions are the solutions to Laplace equations
(where the Laplacian is understood as $\mathrm{div}\circ \mathrm{grad}$.)
In Physics, the construction we have described is known as a sigma model,
and harmonicity plays a very important rôle in their study. In this paper,
we want to show that the relation between harmonicity and the Laplacian
is preserved when passing to the supermanifold setting, and this is
done by applying the methods of the calculus of variations to the
sections of a graded submersion representing supersigma models.

We would like to stress two features of our approach:
\begin{enumerate}
\item We use odd metrics.
\item The superlaplacian $\triangle$ is not trivial, due to the fact that we use
the Berezinian sheaf to define it (compare this with the geometric construction in
\cite{MonSan96}, where any superfunction is harmonic).

\end{enumerate}

To make the paper relatively self-contained, we include two brief sections
on the Berezinian module and graded metrics on supermanifolds. Then we write
our Lagrangian for the supersymmetric sigma model and associate to it
a variational problem, so we can apply the techniques of the graded calculus
of variations (for a detailed account of these techniques, and general definitions,
notations and results on supermanifold theory, see \cite{MonVal03, MMV06} and 
references therein).

\section{$\mathbb{Z}_{2}$-graded metrics}

Let $(M,\mathcal{A})=(M,\wedge\mathcal{E})$ be an $(m|n)$-dimensional
supermanifold given in Batchelor's form, so $\mathcal{E}=\Gamma(E)$
where $E\rightarrow M$ is a vector bundle. We denote by
$\{\frac{\partial}{\partial x^{k}}\}_{k=1}^{m}$
a local frame on $M$, and by $\{x^{-j}\}_{j=1}^{n}$ a local frame
of sections of $\mathcal{E}$. By $\sim $ we will understand the structural
morphism of the supermanifold, $\sim \colon \mathcal{A} \to C^{\infty}(M)$.

As in \cite{MonSan97}, a $\mathbb{Z}_{2}$-graded metric (or supermetric)
on $\wedge\mathcal{E}$ is understood to be a graded symmetric, non
degenerate $\wedge\mathcal{E}$-bilinear map
$G\colon\mathrm{Der}\wedge\mathcal{E}\times\mathrm{Der}
\wedge\mathcal{E}\to \wedge\mathcal{E}$,
whose action on a pair $(D_{1},D_{2})\in\mathrm{Der}
\wedge\mathcal{E}\times\mathrm{Der}\wedge\mathcal{E}$
is denoted $\left\langle D_{1},D_{2};G\right\rangle $ (for another approach to
supermetrics, based on principal superbundles, see \cite{Sar08}.) We have

\begin{enumerate}
\item $\left\langle \alpha D_{1},D_{2};G\right\rangle
=\alpha\left\langle D_{1},D_{2};G\right\rangle $,
$\forall\alpha\in\wedge\mathcal{E}$.
\item $\left\langle D_{1},D_{2};G\right\rangle
=(-1)^{|D_{1}||D_{2}|}\left\langle D_{2},D_{1};G\right\rangle $.
\item The map
$G^{\flat}\colon D\mapsto\left\langle D,.;G\right\rangle $
is an isomorphism between the $\wedge\mathcal{E}$-modules
$\mathrm{Der}\wedge\mathcal{E}$ and
$\mathrm{Hom}(\mathrm{Der}\wedge\mathcal{E},\wedge\mathcal{E})
=\Omega_{G}^{1}(M)$.
The inverse of this isomorphism is denoted by $G^{\sharp}$, as in
the non-graded setting. Also, the induced metric on $\Omega_{G}^{1}(M)$
is denoted by $G^{-1}$.
\end{enumerate}
Note that these conditions imply
$\left\langle D_{1},\alpha D_{2};G\right\rangle
=(-1)^{|\alpha||D_{1}|}\left\langle D_{1},D_{2};G\right\rangle $.
Also, we remark that graded forms have a right $\wedge\mathcal{E}$-module
structure, so we will take care in writting the $\wedge\mathcal{E}$
factors to the left of graded vector fields and to the right of graded
forms.

A graded metric is even (resp.\ odd) if,
$\left\vert \left\langle D_{1},D_{2};G\right\rangle \right\vert
+\left\vert D_{1}\right\vert +\left\vert D_{2}\right\vert \equiv 0\mathrm{mod}2$
(resp.\ $\equiv 1\mathrm{mod}2$). In both cases the graded metric is called
homogeneous. The associated concept of a graded connection is defined
analogously to the non-graded case.

We focus our attention to graded metrics of second-order depth. This
means that there is a graded basis (a set of generators of the locally
finite $\mathcal{\wedge E}$-module $\mathrm{Der}\wedge\mathcal{E}$)
$\{D_{\alpha}\}_{\alpha=1}^{m+n}$ such that,
\[
\left\langle D_{i},D_{j};G\right\rangle \in \sum\limits _{0\leq k\leq2}
\wedge^{k}\mathcal{E}.
\]
In \cite[Proposition 4.1]{MonSan97} the following is proved:\\
Let $\left\langle \cdot,\cdot\right\rangle $ be a homogeneous graded
metric of second order depth which is adapted to the canonical splitting
of $\wedge\mathcal{E}$. A connection $\nabla$ exists on $M$ such
that,

\begin{equation}
\langle \cdot,\cdot;G\rangle =
 \begin{cases}
  \left( 
   \begin{array}{cc}
   g & 0\\
   0 & \omega
   \end{array}
  \right)& \; with \; \nabla g=0 \; (even \; case)\\
  \left( 
   \begin{array}{cc}
   0 & \kappa\\
   \kappa^{t} & 0
   \end{array} \right)& \; (odd \; case,)
 \end{cases}
\end{equation}
with respect to the basis
$\left\{ \nabla_{\frac{\partial}{\partial x^{k}}},
i_{\frac{\partial}{\partial x^{-j}}}\right\} _{1\leq k\leq m}^{1\leq j\leq n}$
of $\mathrm{Der}\wedge\mathcal{E}$, where $g$ is a metric on $TM$, $\omega $
is a symplectic form on $M$, and
$\kappa\colon\mathcal{X}(M)\to \mathcal{E}^{\ast}$
is a non-degenerate linear map.

If $d^{G}$ denotes the graded exterior derivative acting on the algebra
of graded forms $\Omega_{G}(M)$---in order to distinguish it from
the usual exterior derivative $d$, which acts on $\Omega(M)$---and
$\{d^{G}x^{k},d^{G}x^{-j}\} _{1\leq k\leq m}^{1\leq j\leq n}$ is the dual
basis on $\Omega_{G}^{1}(M)$ to the basis of derivations
$\left\{ \nabla_{\frac{\partial}{\partial x^{k}}},
i_{\frac{\partial}{\partial x^{-j}}}\right\} _{1\leq k\leq m}^{1\leq j\leq n}$,
we can write
\[
\left\langle .,.;G\right\rangle =d^{G}x^{i}\otimes d^{G}x^{-j}\cdot\kappa_{ij}
+d^{G}x^{-i}\otimes d^{G}x^{j}\cdot\kappa_{ji}.\]
 for any odd metric on $(M,\wedge\mathcal{E})$. Note that, in the
case $\mathcal{E}=\Gamma(TM)$, we can make the identifications
$\nabla_{\frac{\partial}{\partial x^{k}}}\equiv\frac{\partial}{\partial x^{k}}$
and $i_{\frac{\partial}{\partial x^{-j}}}=\frac{\partial}{\partial x^{-j}}$
as differential operators on the $C^{\infty}(M)$-algebra $\wedge\mathcal{E}$.

\begin{example}
\label{example1}
Consider the linear supermanifold
$\mathbb{R}^{1|1}=(\mathbb{R},\Omega(\mathbb{R}))$ with global supercoordinates
$\{t,\tau\}$, $|t|=0$, $|\tau|=1$. The canonical odd Euclidean
supermetric can be written as
\begin{equation}
Q=d^{G}t\otimes d^{G}\tau+d^{G}\tau\otimes d^{G}t.\label{hmetric}
\end{equation}
 Note that $\mathbb{R}^{1|1}$ admits no even supermetrics (this would
require the base manifold being even dimensional). \end{example}

\begin{example}
\label{example2} Let $\mathcal{E}=\Gamma(TM)$ and
let $g$ be a metric on $M$. A second-order depth odd supermetric
$G$ on $(M,\wedge\mathcal{E})$ can naturally be defined by simply
taking as $\kappa$ the induced isomorphism
$g\colon\mathcal{X}(M)\to \Omega^{1}(M)$.
In local coordinates:
\begin{equation}
G=d^{G}x^{i}\otimes d^{G}x^{-j}\cdot g_{ij}
+d^{G}x^{-i}\otimes d^{G}x^{j}\cdot g_{ji},
\label{gmetric}
\end{equation}
 where $g_{ij}$ is the matrix of $g$ with respect to the local frame
$\{\frac{\partial}{\partial x^{k}}\}_{k=1}^{m}$ on $M$.
\end{example}

\section{The Berezinian sheaf and divergence}\label{Bersheaf}

Let $(M,\mathcal{A})$ be a graded manifold, of dimension $(m|n)$, and let
$\mathrm{P}^{k}(\mathcal{A})$ be the sheaf of graded $k-$order differential
operators of $\mathcal{A}$. This is the submodule
of $\mathrm{End}(\mathcal{A})$ whose elements $P$ verify
\[
\lbrack...[[P,a_{0}],a_{1}],...,a_{k}]=0,
\]
for all $a_{0},...,a_{k}\in\mathcal{A}$ (here we identify an $a\in\mathcal{A}$
with the endomorphism $b\mapsto ab$).

One has that if $\{x^{i},x^{-j}\}_{1\leq j\leq n}^{1\leq i\leq m}$ are
supercoordinates for a splitting neighborhood $U\subset M$,
$\mathrm{P}^{k}(\mathcal{A}(U))$ is a free module (for both structures,
left and right) with basis
{\footnotesize
\[
\frac{\partial^{|\alpha|}}{\partial x^{^{\alpha}}}\circ\frac{\partial
^{|\beta|}}{\partial x^{-\beta}}
=\left(
\frac{\partial}{\partial x^{1}}\right) ^{\alpha_{1}}
\circ\cdots\circ
\left(
\frac{\partial }{\partial x^{m}}\right) ^{\alpha_{m}}\circ
\left( \frac{\partial}{\partial x^{-1}}\right) ^{\beta_{1}}
\circ \cdots \circ
\left(  \frac{\partial }{\partial x^{-n}}\right) ^{\beta_{n}},
\]
} where $|\alpha|+|\beta|\leq k$.

Let us consider the sheaf $\mathrm{P}^{k}(\mathcal{A},\Omega_{G}^{m})
=\Omega_{G}^{m}\otimes_{\mathcal{A}}\mathrm{P}^{k}(\mathcal{A})$,
of $m-$form valued $k-$th order differential operators of $\mathcal{A}$.
For every open subset $U\subset M$, let $\mathrm{K}^{n}(U)$ be the set
of operators $P\in\mathrm{P}^{n}(\mathcal{A}(U),\Omega_{G}^{m}(U))$
such that for every $a\in\mathcal{A}(U)$ with compact support, there exists
an ordinary $(m-1)-$form of compact support, $\omega $, fulfilling
$\widetilde{P(a)}=d\omega$. We observe that $\mathrm{K}^{n}$
is a submodule of $\mathrm{P}^{n}(\mathcal{A},\Omega_{G}^{m})$
for its\emph{\ right structure}, so we can take quotients
and obtain the following description of the Berezinian sheaf:
\begin{equation}
\mathrm{Ber}(\mathcal{A})
=\mathrm{P}^{n}(\mathcal{A},\Omega_{G}^{m}) \diagup \mathrm{K}^{n}.
\end{equation}

According to this description, a local basis of $\mathrm{Ber}(\mathcal{A})$
can be given explicitly: If $\{x^{i},x^{-j}\}_{1\leq j\leq n}^{1\leq i\leq m}$
are supercoordinates for a splitting neighborhood $U\subset M$, the local
sections of the Berezinian sheaf are written in the form
\begin{equation}
\Gamma_{U}(\mathrm{Ber}(\mathcal{A}))
=\left[
d^{G}x^{1}\wedge \cdots \wedge d^{G}x^{m}\otimes
\frac{\partial}{\partial x^{-1}}\circ\cdots \circ
\frac{\partial}{\partial x^{-n}}
\right]
\cdot\mathcal{A}(U),
\end{equation}
where $[$ $]$ stands for the equivalence class modulo $\mathrm{K}^{n}$.

Now, if $X$ is a graded vector field, it is possible to define the notion
of graded Lie derivative of sections of the Berezinian sheaf with respect
to $X$. This is the mapping
\begin{equation}
\mathcal{L}_{X}^{G}\colon \Gamma(\mathrm{Ber}(\mathcal{A}))
\longrightarrow \Gamma(\mathrm{Ber}(\mathcal{A}))
\end{equation}
given by
\begin{equation}
\mathcal{L}_{X}^{G}[\eta^{G}\otimes P]
=(-1)^{|X||\eta^{G}\otimes P|+1}[\eta^{G}\otimes P\circ X],
\end{equation}
for $X$ and $\eta^{G}\otimes P$ homogeneous.

This Lie derivative, has the properties that one would expect (cfr. the treatment
in \cite{QFT99}, Vol. 1 pg. $83$):

\begin{enumerate}
\item For homogeneous $X\in\mathrm{Der}(\mathcal{A})$,
$\xi\in\Gamma (\mathrm{Ber}(\mathcal{A}))$ and $a\in\mathcal{A}$,
\[
\mathcal{L}_{X}^{G}(\xi\cdot a)
=\mathcal{L}_{X}^{G}(\xi)\cdot a+(-1)^{|X||\xi |}\xi\cdot X(a).
\]

\item For homogeneous $X\in\mathrm{Der}(\mathcal{A})$,
$\xi \in \Gamma (\mathrm{Ber}(\mathcal{A}))$ and $a\in\mathcal{A}$,
\[
\mathcal{L}_{a\cdot X}^{G}(\xi)
=(-1)^{|X||\xi|}\mathcal{L}_{X}^{G}(\xi\cdot a).
\]

\item If $\xi_{x^{i},x^{-j}}=[d^{G}x^{1}\wedge \cdots
\wedge d^{G}x^{m}\otimes\frac{\partial}{\partial x^{-1}}
\circ \cdots \circ \frac{\partial}{\partial x^{-n}}]$
is the local generator of the Berezinian sheaf
on a system of supercoordinates
$\{x^{i},x^{-j}\}_{1\leq j\leq n}^{1\leq i\leq m}$, then
\[
\mathcal{L}_{\frac{\partial}{\partial x^{i}}}^{G}(\xi_{x^{i},x^{-j}})
=0=\mathcal{L}_{\frac{\partial}{\partial x^{-j}}}^{G}(\xi_{x^{i},x^{-j}}).
\]

\end{enumerate}

We can now introduce the notion of Berezinian divergence: Let $(M,\mathcal{A})$
be a graded manifold whose Berezinian sheaf is generated
by a section $\xi$. The graded function $\mathrm{div}_{B}^{\xi}(X)$ given by
the formula (for homogeneous $X$)
\[
\mathcal{L}_{X}^{G}(\xi)=(-1)^{|X||\xi|}\xi\cdot\mathrm{div}_{B}^{\xi}(X)
\]
is called the Berezinian divergence of $X$ with respect to $\xi$. When there
is no risk of confusion, we will write simply $\mathrm{div}_{B}(X)$.

\section{$(1|1)$-supersymmetric sigma model}

Below, we consider a model with target $\mathbb{R}^{1|1}$, that is,
a mapping (scalar superfield)
$\sigma \colon(M,\wedge\mathcal{E})\to \mathbb{R}^{1|1}$,
where we consider an odd supermetric on $(M,\wedge\mathcal{E})$ as
in Example \ref{example2} and the canonical metric of Example \ref{example1}
on $\mathbb{R}^{1|1}$.

This imply that $\mathcal{E}\cong\Gamma(TM)$, but this is not a great
loss of generality: If $(M,\wedge\mathcal{E})$ admits an odd metric,
the existence of the non-degenerate pairing
$\kappa\colon\mathcal{X}(M)\rightarrow\mathcal{E}^{\ast}$
implies that the dimension of the supermanifold is $(n|n)$.

The mapping $\sigma$ can be viewed as a section of the graded submersion
\[
p\colon\mathbb{R}^{1|1}\times(M,\wedge\mathcal{E})\to (M,\wedge\mathcal{E}),
\]
to which we associate a super-Lagrangian $L\in\mathcal{A}_{J_{G}^{1}(p)}$
(see \cite{MMV06} for the details of this construction) proceeding by
analogy with the non-graded case (which is that of harmonic functions).
The coordinates in the bundle of superjets $J_{G}^{1}(p)$ will be denoted
$\{t,\tau,x^{i},x^{-j},t_{i},t_{-j},\tau_{i},\tau_{-j}\}$.
Thus, by taking the pull-back (we write $x^{\alpha}$ collectively
for $x^{i}$ and $x^{-j}$, i.e., $\alpha$ runs from $-n,...,-1$
to $1,...,m$) of (\ref{hmetric}) along $\sigma$, we obtain
\begin{eqnarray*}
\sigma^{\ast}Q & = & (-1)^{\alpha\beta}d^{G}x^{\alpha}
\otimes d^{G}x^{\beta}\cdot\frac{\partial(t\circ\sigma)}{\partial x^{\alpha}}
\frac{\partial(\tau\circ\sigma)}{\partial x^{\beta}}\\
& + & (-1)^{(\gamma+1)\delta}d^{G}x^{\gamma}\otimes d^{G}x^{\delta}\cdot
\frac{\partial(\tau\circ\sigma)}{\partial x^{\gamma}}
\frac{\partial(t\circ\sigma)}{\partial x^{\delta}}.
\end{eqnarray*}
As $|-j|=1$, $|k|=0$, and so on, we could write
matricially\footnote{Upper left corner
corresponds to the factor of $d^{G}x^{i}\otimes d^{G}x^{k}$,
upper right to $d^{G}x^{j}\otimes d^{G}x^{-l}$, lower left
to $d^{G}x^{-j}\otimes d^{G}x^{k}$,
and lower right to $d^{G}x^{-j}\otimes d^{G}x^{-l}$.
},
\[
\sigma^{\ast}Q=\! \left( \;
\begin{array}{cc}
\frac{\partial(t\circ\sigma)}{\partial x^{i}}
\frac{\partial(\tau\circ\sigma)}{\partial x^{k}}
+\frac{\partial(\tau\circ\sigma)}{\partial x^{i}}
\frac{\partial(t\circ\sigma)}{\partial x^{k}}
& \!\! \frac{\partial(t\circ\sigma)}{\partial x^{j}}
\frac{\partial(\tau\circ\sigma)}{\partial x^{-l}}
-\frac{\partial(\tau\circ\sigma)}{\partial x^{j}}
\frac{\partial(t\circ\sigma)}{\partial x^{-l}}\\
\frac{\partial(t\circ\sigma)}{\partial x^{-j}}
\frac{\partial(\tau\circ\sigma)}{\partial x^{k}}
+\frac{\partial(\tau\circ\sigma)}{\partial x^{-j}}
\frac{\partial(t\circ\sigma)}{\partial x^{k}}
& \!\! -\frac{\partial(t\circ\sigma)}{\partial x^{-j}}
\frac{\partial(\tau\circ\sigma)}{\partial x^{-l}}
+\frac{\partial(\tau\circ\sigma)}{\partial x^{-j}}
\frac{\partial(t\circ\sigma)}{\partial x^{-l}}
\end{array}
\; \right).
\]
 Also, from (\ref{gmetric}) we have
 $G^{-1}=g^{ij}\cdot\frac{\partial}{\partial x^{i}}\otimes
 \frac{\partial}{\partial x^{-j}}+g^{kl}\cdot
 \frac{\partial}{\partial x^{k}}\otimes\frac{\partial}{\partial x^{l}}$,
or equivalently, {\tiny $G^{-1}=\left(\!
\begin{array}{cc}
\!0\! & \! g^{ij}\!\\
\! g^{ij}\! & \!0\!\end{array}
\!\right)$}.

We would like to evaluate the action of $G^{-1}$ on $\sigma^{\ast}Q$
in such a way that the functorial correspondence with the composition
of graded morphisms be preserved. To this end, we must use either
the action of the metric $G^{-1}$ on $\sigma^{\ast}Q$ or the notion
of matrix product developed in \cite{San88}; the result is the following:
{\small
\begin{eqnarray*}
G^{-1}\cdot \sigma^{\ast}Q = & & \\
\left( \; \begin{array}{cc}
g^{ij}\left(\frac{\partial(t\circ\sigma)}{\partial x^{-j}}
\frac{\partial(\tau\circ\sigma)}{\partial x^{k}}
+\frac{\partial(\tau\circ\sigma)}{\partial x^{-j}}
\frac{\partial(t\circ\sigma)}{\partial x^{k}}\right)
& \!\! g^{ij}\left(-\frac{\partial(t\circ\sigma)}{\partial x^{-j}}
\frac{\partial(\tau\circ\sigma)}{\partial x^{-l}}
+\frac{\partial(\tau\circ\sigma)}{\partial x^{-j}}
\frac{\partial(t\circ\sigma)}{\partial x^{-l}}\right)\\
g^{ij}\left(\frac{\partial(t\circ\sigma)}{\partial x^{i}}
\frac{\partial(\tau\circ\sigma)}{\partial x^{k}}
+\frac{\partial(\tau\circ\sigma)}{\partial x^{i}}
\frac{\partial(t\circ\sigma)}{\partial x^{k}}\right)
& \!\! g^{ij}\left(\frac{\partial(t\circ\sigma)}{\partial x^{j}}
\frac{\partial(\tau\circ\sigma)}{\partial x^{-l}}
-\frac{\partial(\tau\circ\sigma)}{\partial x^{j}}
\frac{\partial(t\circ\sigma)}{\partial x^{-l}}
\right)
\end{array}
\; \right) \!.
\end{eqnarray*}
}

If we were to take the supertrace of this supermatrix, we would obtain
$0$, as the following computation shows:
\begin{eqnarray*}
Str(G^{-1}\cdot\sigma^{\ast}Q)
& = & Tr\left(g^{ij}\left( \frac{\partial(t\circ\sigma)}{\partial x^{-j}}
\frac{\partial(\tau\circ\sigma)}{\partial x^{k}}
+\frac{\partial(\tau\circ\sigma)}{\partial x^{-j}}
\frac{\partial(t\circ\sigma)}{\partial x^{k}}\right)\right)\\
& - & Tr\left(g^{ij}\left(\frac{\partial(t\circ\sigma)}{\partial x^{j}}
\frac{\partial(\tau\circ\sigma)}{\partial x^{-l}}
-\frac{\partial(\tau\circ\sigma)}{\partial x^{j}}
\frac{\partial(t\circ\sigma)}{\partial x^{-l}}\right)\right)\\
& = & g^{ij}\left(\frac{\partial(t\circ\sigma)}{\partial x^{-j}}
\frac{\partial(\tau\circ\sigma)}{\partial x^{i}}
+\frac{\partial(\tau\circ\sigma)}{\partial x^{-j}}
\frac{\partial(t\circ\sigma)}{\partial x^{i}}\right)\\
& - & g^{kl}\left(\frac{\partial(t\circ\sigma)}{\partial x^{l}}
\frac{\partial(\tau\circ\sigma)}{\partial x^{-k}}
-\frac{\partial(\tau\circ\sigma)}{\partial x^{l}}
\frac{\partial(t\circ\sigma)}{\partial x^{-k}}\right)\\
& = & g^{ij}\frac{\partial(t\circ\sigma)}{\partial x^{-j}}
\frac{\partial(\tau\circ\sigma)}{\partial x^{i}}+g^{kl}
\frac{\partial(\tau\circ\sigma)}{\partial x^{l}}
\frac{\partial(t\circ\sigma)}{\partial x^{-k}}\\
& + & g^{ij}\frac{\partial(\tau\circ\sigma)}{\partial x^{-j}}
\frac{\partial(t\circ\sigma)}{\partial x^{i}}-g^{kl}
\frac{\partial(t\circ\sigma)}{\partial x^{l}}
\frac{\partial(\tau\circ\sigma)}{\partial x^{-k}}\\
& = & g^{ij}\frac{\partial(t\circ\sigma)}{\partial x^{-j}}
\frac{\partial(\tau\circ\sigma)}{\partial x^{i}}-g^{kl}
\frac{\partial(t\circ\sigma)}{\partial x^{-k}}
\frac{\partial(\tau\circ\sigma)}{\partial x^{l}}\\
& + & g^{ij}\frac{\partial(\tau\circ\sigma)}{\partial x^{-j}}
\frac{\partial(t\circ\sigma)}{\partial x^{i}}-g^{kl}
\frac{\partial(\tau\circ\sigma)}{\partial x^{-k}}
\frac{\partial(t\circ\sigma)}{\partial x^{l}}=0.
\end{eqnarray*}

Due to this, instead of $Str$ we will take the usual contraction
of $(1,1)$-supertensors, which amounts to
\begin{eqnarray*}
C_{1}^{1}(G^{-1}\cdot\sigma^{\ast}Q)
& = & Tr\left( g^{ij}\left( \frac{\partial(t\circ\sigma)}{\partial x^{-j}}
\frac{\partial(\tau\circ\sigma)}{\partial x^{k}}
+\frac{\partial(\tau\circ\sigma)}{\partial x^{-j}}
\frac{\partial(t\circ\sigma)}{\partial x^{k}}\right)\right)\\
& + & Tr\left(g^{ij}\left( \frac{\partial(t\circ\sigma)}{\partial x^{j}}
\frac{\partial(\tau\circ\sigma)}{\partial x^{-l}}
-\frac{\partial(\tau\circ\sigma)}{\partial x^{j}}
\frac{\partial(t\circ\sigma)}{\partial x^{-l}}\right)\right)\\
& = & g^{ij}\left( \frac{\partial(t\circ\sigma)}{\partial x^{-j}}
\frac{\partial(\tau\circ\sigma)}{\partial x^{i}}
+\frac{\partial(\tau\circ\sigma)}{\partial x^{-j}}
\frac{\partial(t\circ\sigma)}{\partial x^{i}}\right)\\
& + & g^{kl}\left(\frac{\partial(t\circ\sigma)}{\partial x^{l}}
\frac{\partial(\tau\circ\sigma)}{\partial x^{-k}}
-\frac{\partial(\tau\circ\sigma)}{\partial x^{l}}
\frac{\partial(t\circ\sigma)}{\partial x^{-k}}\right)\\
& = & g^{ij}\frac{\partial(t\circ\sigma)}{\partial x^{-j}}
\frac{\partial(\tau\circ\sigma)}{\partial x^{i}}-g^{kl}
\frac{\partial(\tau\circ\sigma)}{\partial x^{l}}
\frac{\partial(t\circ\sigma)}{\partial x^{-k}}\\
& + & g^{ij}\frac{\partial(\tau\circ\sigma)}{\partial x^{-j}}
\frac{\partial(t\circ\sigma)}{\partial x^{i}}
+g^{kl}\frac{\partial(t\circ\sigma)}{\partial x^{l}}
\frac{\partial(\tau\circ\sigma)}{\partial x^{-k}}\\
& = & 2g^{ij}\frac{\partial(t\circ\sigma)}{\partial x^{-j}}
\frac{\partial(\tau\circ\sigma)}{\partial x^{i}}+2g^{ij}
\frac{\partial(t\circ\sigma)}{\partial x^{i}}
\frac{\partial(\tau\circ\sigma)}{\partial x^{-j}}.
\end{eqnarray*}
That is, we have arrived at the result
\begin{equation}
\frac{1}{2}C_{1}^{1}(G^{-1}\cdot\sigma^{\ast}Q)
=g^{ij}\left(\frac{\partial(t\circ\sigma)}{\partial x^{-j}}
\frac{\partial(\tau\circ\sigma)}{\partial x^{i}}
+\frac{\partial(t\circ\sigma)}{\partial x^{i}}
\frac{\partial(\tau\circ\sigma)}{\partial x^{-j}}\right) ,
\label{contraction}
\end{equation}
 and this suggests to take the following superlagrangian
 $L\in\mathcal{A}_{J_{G}^{1}(p)}$
to study the supersymmetric sigma model (where we introduce the obvious
notations $t_{\alpha}=\frac{\partial(t\circ\sigma)}{\partial x^{\alpha}}$,
$\tau_{\beta}=\frac{\partial(\tau\circ\sigma)}{\partial x^{\beta}}$):
\begin{equation}
L=g^{ij}(t_{-j}\tau_{i}+t_{i}\tau_{-j}).\label{sigmalag}
\end{equation}

\begin{remark} The superlagrangian $L$ is homogeneous of even degree.
In Physics literature, it is common to write $t\circ\sigma=\phi$
and $\tau\circ\sigma=\psi$, so equation (\ref{contraction}) would
be
\[
\frac{1}{2}C_{1}^{1}(G^{-1}\cdot\sigma^{\ast}Q)
=g^{ij}\left( \frac{\partial\phi}{\partial x^{-j}}
\frac{\partial\psi}{\partial x^{i}}+\frac{\partial\phi}{\partial x^{i}}
\frac{\partial\psi}{\partial x^{-j}}\right)
\]
 and the Lagrangian density
 \begin{equation}
L=g^{ij}(\phi_{-j}\psi_{i}+\phi_{i}\psi_{-j}).\label{sigmalag2}
\end{equation}

\end{remark}

\section{The associated variational problem and harmonic superfunctions}

Next, we want to study the Euler-Lagrange equations corresponding
to the Berezinian problem defined by $L$. To this end, we recall
that $G$ determines a global section $\xi_{G}$ of the Berezinian
sheaf $\mathrm{Ber}(\mathcal{A})=\mathrm{Ber}(\wedge\mathcal{E})$,
which is called the Riemannian Berezinian volume element associated
to $G$. In coordinates:
\[
\xi_{G}=\left[d^{G}x^{1}\wedge\cdots\wedge d^{G}x^{m}\otimes
\frac{\partial}{\partial x^{-1}}\circ\cdots\circ
\frac{\partial}{\partial x^{-n}}\right]|G|
\]
 where $|G|=\sqrt{Ber(G_{\alpha\beta})}$ and $Ber$ is the Berezinian
determinant (note that $|G|$ is even and it only depends on the even
coordinates $x^{i}$). The Euler-Lagrange equations are
\begin{eqnarray*}
\frac{\partial\lambda}{\partial t}-\frac{d}{dx^{i}}
\frac{\partial\lambda}{\partial t_{i}}-\frac{d}{dx^{-j}}
\frac{\partial\lambda}{\partial t_{-j}} & = &
0\mbox{ , }1\leq i\leq m,1\leq j\leq n\\
\frac{\partial\lambda}{\partial\tau}
-\frac{d}{dx^{i}}\frac{\partial\lambda}{\partial\tau_{i}}
+\frac{d}{dx^{-j}}\frac{\partial\lambda}{\partial\tau_{-j}}
& = & 0\mbox{ , }1\leq i\leq m,1\leq j\leq n,
\end{eqnarray*}
 with $\lambda=|G|\cdot L$.

Taking the explicit expression (\ref{sigmalag}) or (\ref{sigmalag2})
into account, it is readily seen that these equations reduce respectively to
\begin{eqnarray}\label{eulero}
\frac{1}{|G|}\frac{\partial|G|}{\partial x^{k}}g^{kj}\tau_{-j}
+\frac{\partial g^{kj}}{\partial x^{k}}\tau_{-j}
+g^{kj}\tau_{k,-j}+g^{kj}\tau_{-kj}
& = & 0 \nonumber \\
\frac{1}{|G|}\frac{\partial|G|}{\partial x^{k}}g^{kj}t_{-j}
+\frac{\partial g^{kj}}{\partial x^{k}}t_{-j}+g^{kj}t_{k,-j}+g^{kj}
& = & 0.
\end{eqnarray}

Classically, the Euler-Lagrange variational equations for the sigma model with
target manifold $\mathbb{R}$ are those of the harmonic functions,
characterized by $\Delta f=0$, where $\Delta$ is the ordinary Laplacian:
$\Delta f=\mathrm{div}(\mathrm{grad}f)$ (see \cite{Raw84,Nis02}). We would like to show that
this is still true in the graded setting, that is, that equations
(\ref{eulero}) when evaluated on sections are the local expression
of the harmonic superfunctions, aside from constant factors. Of course,
to define the super-Laplacian $\Delta$ one must specify first what
is meant by div and grad in a supermanifold.

If $f\in\wedge\mathcal{E}$ is a superfunction on $(M,\wedge\mathcal{E})$
with graded metric $G$, its gradient is defined as the supervector
field given by $ \left\langle \mathrm{grad}f,D\right\rangle =D(f)$,
for all $D\in\mathrm{Der}(\wedge \mathcal{E})$.

\begin{proposition} For a superfunction $f\in\wedge\mathcal{E}$,
the following local expression holds true,
\begin{equation}
\mathrm{grad}f=g^{ij}\frac{\partial f}{\partial x^{-j}}
\frac{\partial}{\partial x^{i}}+g^{kl}\frac{\partial f}{\partial x^{k}}
\frac{\partial}{\partial x^{-l}}.\label{gradient}
\end{equation}
\end{proposition}

\begin{proof} Indeed, making use of the $\wedge\mathcal{E}$-bilinearity
of $G$ and the explicit form (\ref{gmetric}), if
$\mathrm{grad}f=A^{\alpha}\frac{\partial}{\partial x^{\alpha}}$,
then
\[
\frac{\partial f}{\partial x^{-j}}
=\left\langle \mathrm{grad}f,\frac{\partial}{\partial x^{-j}}\right\rangle
=A^{\alpha}\left\langle \frac{\partial}{\partial x^{\alpha}},
\frac{\partial}{\partial x^{-j}}\right\rangle =A^{i}g_{ij}.
\]
 Hence $A^{i}=g^{ij}\frac{\partial f}{\partial x^{-j}}$, and similarly
for $A^{-l}=g^{kl}\frac{\partial f}{\partial x^{k}}$.
\end{proof}

Next, we study the divergence. If we consider the Riemannian Berezinian
$\xi_{G}=\xi|G|$ and compute the divergence
of $D\in\mathrm{Der}(\wedge\mathcal{E})$
with respect to it, we obtain the following (recall the properties
of the Lie derivative of sections of the Berezinian sheaf from section
\ref{Bersheaf}):

\begin{proposition} For any $D\in\mathrm{Der}(\wedge\mathcal{E})$,
we have\begin{equation}
\mathrm{div}D=(-1)^{|D||\xi|}\frac{1}{|G|}
\frac{\partial}{\partial x^{\alpha}}(|G|\cdot D^{\alpha}).
\label{divergence}\end{equation}
\end{proposition}

\begin{proof} It goes as follows:
\begin{eqnarray*}
\mathcal{L}_{D}^{G}\xi_{G}
& = & \mathcal{L}_{D^{\alpha}
\frac{\partial}{\partial x^{\alpha}}}^{G}(\xi\cdot|G|)\\
& = & (-1)^{\alpha|\xi|}
\mathcal{L}_{\frac{\partial}{\partial x^{\alpha}}}^{G}
(\xi\cdot|G|\cdot D^{\alpha})\\
& = & (-1)^{\alpha|\xi|}\left(
\mathcal{L}_{\frac{\partial}{\partial x^{\alpha}}}^{G}
(\xi)\cdot(|G|\cdot D^{\alpha})
+(-1)^{\alpha|\xi |}\xi \cdot
\frac{\partial}{\partial x^{\alpha}}(|G|\cdot D^{\alpha})
\right)\\
& = & \xi\cdot\frac{\partial}{\partial x^{\alpha}}
(|G|\cdot D^{\alpha})\\
& = & \xi\cdot|G|\cdot\frac{1}{|G|}
\frac{\partial}{\partial x^{\alpha}}(|G|\cdot D^{\alpha})\\
& = & \xi_{G} \cdot\frac{1}{|G|}
\frac{\partial}{\partial x^{\alpha}}(|G|\cdot D^{\alpha}).
 \end{eqnarray*}
\end{proof}

Finally, from (\ref{gradient}) and (\ref{divergence}), we obtain
the expression for the super-Laplacian of a superfunction
$f\in\wedge\mathcal{E}$,
namely,

\begin{proposition}
For any superfunction $f\in\wedge\mathcal{E}$,
the equation $\Delta(f)=0$ is equivalent to:
\[
\frac{1}{|G|}\frac{\partial|G|}{\partial x^{i}}g^{ij}f_{-j}
+\frac{\partial g^{ij}}{\partial x^{i}}f_{-j}+g^{ij}f_{i,-j}+g^{ij}f_{-i,j}=0.
\]
\end{proposition}

\begin{proof} By a direct computation:
{\small
\begin{eqnarray*}
(-1)^{(|f|+1)|\xi|}\Delta(f)
& = & (-1)^{(|f|+1)|\xi|}\mathrm{div}(\mathrm{grad}f)\\
& = & \frac{1}{|G|}\frac{\partial}{\partial x^{\alpha}}
(|G|\cdot(\mathrm{grad}f)^{\alpha})\\
& = & \frac{1}{|G|}\left( \frac{\partial}{\partial x^{i}}
\left(|G|\cdot g^{ij}\frac{\partial f}{\partial x^{-j}}\right)
+\frac{\partial}{\partial x^{-i}}\left(|G|\cdot g^{ik}
\frac{\partial f}{\partial x^{k}}\right) \right) \\
& = & \frac{1}{|G|}\frac{\partial|G|}{\partial x^{i}}g^{ij}f_{-j}
+\frac{\partial g^{ij}}{\partial x^{i}}f_{-j}+g^{ij}f_{i,-j}+g^{ij}f_{-i,j}.
\end{eqnarray*}
}
 \end{proof}

Separating the even and odd factors, we obtain (\ref{eulero}) precisely
when evaluated on sections of
$p\colon \mathbb{R}^{1|1}\times(M,\wedge \mathcal{E})\to (M,\wedge\mathcal{E})$.
We have thus achieved the following characterization:

\begin{theorem}
The harmonic superfunctions are precisely the solutions
to the Euler-Lagrange equations of the $(1|1)$-supersymmetric sigma
model.
\end{theorem} 

\section*{Acknowledgments}

Partially supported by the Ministerio
de Educaci\'on y Ciencia of Spain, under grants MTM$2008-01386$ and 
MTM$2005-04947$. JAV also supported by a FAI-UASLP grant C$07$-FAI-$04-18.20$ and
a SEP-CONACyT project CB (J2) 2007-1 code 78791.

\appendix


\begin{thebibliography}{0}

\bibitem{Aoy94}S. Aoyama, \emph{Quantization of the topological
$\sigma-$model and the master equation of the BV formalism.\/}
Modern Phys.\ Lett.\ \textbf{A9} $(1994)$ no.\ 6, $491-500$.

\bibitem{BatBer08}I. A. Batalin, K. Bering,
\emph{Odd scalar curvature in anti-Poisson
geometry}.\ arXiv hep-th:$0712.3699v2$ ($2$ January $2008$).

\bibitem{QFT99}P. Deligne et al. (Eds.), \emph{Quantum Fields and Strings: 
A Course For Mathematicians}, 2 vols., American Mathematical Society, Providence, 
$1999$.

\bibitem{Eel78}J. Eells, L. Lemaire, \emph{A report on harmonic
maps}. Bull. London Math. Soc. \textbf{10} (1978) 1.

\bibitem{Eel64}J. Eells, J. H. Sampson, \emph{Harmonic mappings
of Riemannian manifolds}. Amer. J. Math. \textbf{86} (1964), 109.

\bibitem{KV02}H. Khudaverdian, T. Voronov, \emph{On odd Laplace operators.\/}
Lett.\ Math.\ Phys.\ \textbf{62} $(2002)$ no.\ 2, $127-142$.

\bibitem{KV041}H. Khudaverdian, T. Voronov,
\emph{On odd Laplace operators II}.\ Geometry,
Topology and Mathematical Physics, $179-205$. Amer.\ Math.\ Soc.\
Transl.\ Ser.\ $2$ \textbf{212}.\ Amer.\ Math. Soc., Providence (RI)
$2004$.

\bibitem{KV042}H. Khudaverdian, T. Voronov, \emph{Geometry of differential operators,
odd Laplacians, and homotopy algebras.\/}
J. Nonlinear Math.\ Phys.\ \textbf{11} $(2004)$
suppl., $217-227$.

\bibitem{KosMon02}Y. Kosmann-Schwarzbach, J. Monterde,
\emph{Divergence operators and odd Poisson brackets.\/}
Ann.\ Inst.\ Fourier \textbf{52} $(2002)$ no.\ 2, $419-456$.

\bibitem{MonSan96}J. Monterde and O. A. Sánchez-Valenzuela, \emph{The exterior 
derivative as a Killing vector field}, Israel J. of Math. \textbf{93} $(1996)$, 
$157-170$.

\bibitem{MonSan97}J. Monterde, O. A. Sánchez-Valenzuela, \emph{Graded
metrics adapted to splittings}. Israel J. Math. \textbf{99} $(1997)$,
$231--270$.

\bibitem{MonVal03}J. Monterde, J. A. Vallejo,
\emph{The symplectic structure of Euler-Lagrange superequations
and Batalin-Vilkoviski formalism.\/} J. Phys.\ A: Math.\ and
General \textbf{A36} $(2003)$ $4993-5009$.

\bibitem{MMV06}J. Monterde, J. Muñoz-Masqué, J. A. Vallejo, \emph{The
Poincaré-Cartan form in superfield theory.\/} Int.\ J. Geom.\ Methods
Mod.\ Phys.\ \textbf{3} $(2006)$ no.\ 4, $775-822$.

\bibitem{MV09}J. Muñoz-Masqué, J.A. Vallejo, \emph{Harmonicity in supermanifolds
and sigma models.} To appear in Proceedings of the XVII Fall Workshop on Geometry
and Physics.

\bibitem{Nis02}S. Nishikawa, \emph{Variational problems in geometry}. Translations of
mathematical monographs {\bf 205} American Mathematical Society. Providence, RI (2002). 

\bibitem{Raw84}J. Rawnsley, \emph{Noether's theorem for harmonic maps}. In \emph{
Differential geometric methods in mathematical physics}. S. Sternberg, editor.
D. Reidel, Dordrecht (1984).

\bibitem{San88}O. A. Sánchez-Valenzuela, \emph{Matrix computations
in linear superalgebra}. Linear Algebra Appl.\ \textbf{111} $(1988)$,
$151-181$.
\bibitem{Sar08}G. Sardanashvily, \emph{Supermetrics on supermanifolds}.
Int. J. Geom. Methods Mod. Phys. {\bf v5} $(2008)$ $271-286$.

\bibitem{Sch93}A. Schwarz, \emph{Geometry of Batalin-Vilkovisky quantization.\/}
Comm.\ Math.\ Phys.\ \textbf{155} $(1993)$ no.\ 2, $249-260$.

\bibitem{Wit90}E. Witten, \emph{A note on the antibracket formalism.\/}
Mod.\ Phys.\ Lett.\ A \textbf{5} no.\ 7 (1990) $487-494$.

\end{thebibliography}
\end{document}